\documentclass[a4paper,11pt]{article}

\usepackage{amsmath,amssymb,amsthm}
\newtheorem{definition}{Definition}
\newtheorem{proposition}{\textbf{Proposition}}
\newtheorem{lemma}{Lemma}
\newtheorem{theorem}{Theorem}

\usepackage{fullpage}

\usepackage{enumerate}

\usepackage{xspace}

\newcommand{\cclass}[1]{\ensuremath{\mbox{\textup{#1}}}\xspace}

\newcommand{\NP}{\cclass{NP}}

\newcommand{\p}{\text{poly}}

\newcommand{\N}{\mathbb{N}}

\newcommand{\C}{\ensuremath{\mathcal C}\xspace}
\newcommand{\F}{\ensuremath{\mathbb F}\xspace}

\newcommand{\Q}{\ensuremath{\mathcal{Q}}\xspace}

\newcommand{\perm}[1]{\ensuremath{\text{perm}(#1)}}

\newcommand{\Oh}{\mathcal{O}}

\title{Abusing the Tutte Matrix: \\An Algebraic Instance Compression for the K-set-cycle Problem}
\author{Magnus Wahlstr\"om\\
Max-Planck-Institut f\"ur Informatik,
Saarbr\"ucken, Germany \\
\texttt{wahl@mpi-inf.mpg.de}}

\begin{document}

\maketitle

\begin{abstract}
We give an algebraic, determinant-based algorithm for the~\textsc{$K$-Cycle} problem,
i.e., the problem of finding a cycle through a set of specified elements. Our approach
gives a simple FPT algorithm for the problem, matching the~$\Oh^*(2^{|K|})$ running time
of the algorithm of Bj\"orklund et al.~(SODA, 2012). Furthermore, our approach is open for
treatment by classical algebraic tools (e.g., Gaussian elimination), and we show that it
leads to a \emph{polynomial compression} of the problem, i.e., a polynomial-time reduction
of the~\textsc{$K$-Cycle} problem into an algebraic problem with coding size~$\Oh(|K|^3)$.
This is surprising, as several related problems (e.g.,~\textsc{$k$-Cycle} and the
\textsc{Disjoint Paths} problem) are known not to admit such a reduction unless the
polynomial hierarchy collapses. Furthermore, despite the result, we are not aware of
any \emph{witness} for the~\textsc{$K$-Cycle} problem of size polynomial in~$|K|+\log n$,
which seems (for now) to separate the notions of polynomial compression and polynomial
kernelization (as a polynomial kernelization for a problem in \NP necessarily implies a
small witness). 
\end{abstract}

\section{Introduction}

Parameterized complexity~\cite{DowneyF98,FlumG06} is one of the major approaches for
dealing with \NP-hard problems. In this setting, the input is associated with a
\emph{parameter}~$k$, usually (but not exclusively) either a parameter related to the
solution size, or a structural parameter such as treewidth; the fundamental assumption is that
problems with a smaller parameter value will be easier than general instances. 
The critical notion is that of an \emph{FPT} algorithm, which runs in time~$f(k)\cdot \p(n)$
for some~$f(k)$ where~$\p(n)$ is independent of~$k$, i.e., the combinatorial explosion is
confined to the parameter~$k$. This notion has lead to a large number of interesting
algorithmic principles; for some surveys, see, e.g., the Festschrift of Mike
Fellows~\cite{FellowsFestschrift}. 

One of the most vibrant parts of parameterized complexity in recent years is the subfield
of \emph{kernelization}. A kernelization is one of the basic approaches for creating FPT
algorithms: It is an algorithm which runs in time polynomial in both~$k$ and~$n$, which
reduces the size of the instance (e.g., via reduction rules such as ``remove a vertex
shown not to be required by the solution'') if the size is larger than some~$f(k)$. 
Additionally, beyond being a design paradigm for FPT algorithms, it has been observed that the notion
can be a good way to formalize effective \emph{instance simplification}, e.g.,
preprocessing with a performance guarantee. A \emph{polynomial kernel}, then, is a
polynomial-time procedure which takes an input instance, with parameter~$k$, and produces
an output instance of size at most~$\p(k)$, regardless of the value of~$n$, without
changing the problem status. Great interest has been taken in recent years in the question
of which problems (and which problem parameterizations) admit polynomial kernels. 
This was sparked by the creation of a lower bounds framework by Bodlaender et
al.~\cite{BodlaenderDFH09} and Fortnow and Santhanam~\cite{FortnowS11}. These results
provided a way to exclude the existence of a polynomial kernel, under the hypothesis that
the polynomial hierarchy does not collapse. Later refinements and applications of this
framework can be found in,
e.g.,~\cite{DellvM10,BodlaenderJK11,DellM12,HermelinW12,Drucker12,DomLS09,CyganKPPW12}.  
Significant progress has also been made on the positive side; for a few examples,
see~\cite{BodlaenderFLPST09,Thomasse10,FominLST10}. A recent trend, relevant to the
current paper, is the application of \emph{algebraic} tools to kernelization,
e.g.,~\cite{KratschW12a,KratschW12b}. (See related work, below.)

Sometimes, the results found by these investigations can be quite surprising. As an
example, consider the problems \textsc{Vertex Cover} (find a set of at most~$k$ vertices
in a graph which covers all edges, i.e., a \emph{vertex cover} of size at most~$k$) and
\textsc{Connected Vertex Cover} (find a vertex cover of size at most~$k$ which
additionally is connected). The former is one of the most well-studied problems in
theoretical computer science. In terms of parameterized complexity, it can be solved in
time~$\Oh^*(2^k)$ by a very simple algorithm, and in time~$\Oh^*(1.28^k)$ by more
involved means~\cite{ChenKX10}. It has a simple~$2$-approximation, and a kernel of~$2k$
vertices by the famous Nemhauser-Trotter theorem~\cite{NemhauserT75}. On the other hand,
if the vertex cover is required to be connected, then the problem still has a simple
greedy~$2$-approximation, an~$\Oh^*(2^k)$-time FPT algorithm~\cite{CyganNPPvRW11}, and, as
shown by Dom et al.~\cite{DomLS09}, no polynomial kernel unless the polynomial hierarchy
collapses.

As another example, consider the following three problems. Given a graph~$G$, find (a) a
cycle with at least~$k$ vertices (the \textsc{$k$-Cycle} problem); (b) a cycle passing
through every element of a given set~$K$,~$|K|=k$ (the \textsc{$K$-Cycle} problem); or (c)
a cycle passing through every element of~$K$, which furthermore passes the elements in a
specified order (which we may dub the \textsc{Ordered $K$-Cycle} problem). Which of these
seem more or less general? Which, if any, seems most likely do admit efficient instance
simplification? 

Let us make a quick review of known FPT and kernelization results for these problems.
All are \NP-hard; in the first two problems, setting~$k=n$ yields the \textsc{Hamiltonian
  Cycle} problem. The \textsc{$k$-Cycle} problem is closely related to \textsc{$k$-Path}
(the problem of finding a path of length at least~$k$), and there is by now a variety of
interesting techniques that can be used to solve it in~$2^{\Oh(k)}\p(n)$ time, from the
seminal color-coding technique of Alon et al.~\cite{AlonYZ95}, via the multilinear
detection of Koutis~\cite{Koutis08} (see also Williams~\cite{Williams09}), to the
recent~$\Oh^*(1.66^n)$-time \textsc{Hamiltonian Cycle} algorithm of Bj\"orklund~\cite{Bjorklund10b}, 
which was adapted to a parameterized setting in~\cite{BjorklundHKK10}.  \textsc{Ordered
  $K$-Cycle} is equivalent to the well-known problem \textsc{Disjoint Paths}, where the
input is~$k$ pairs of vertices~$(s_i, t_i)$, and the question is if we can connect all
pairs with pairwise vertex-disjoint paths. This problem seems much more challenging. 
Robertson and Seymour, in the context of the graph minors programme, showed that it is
FPT; more specifically, that it can be solved in time~$\Oh(n^3)$ for every
fixed~$k$~\cite{RobertsonS95}. Kawarabayashi et al.~improved this to~$\Oh(n^2)$ for every
fixed~$k$~\cite{KawarabayashiKR12}. However, the algorithms are in both cases very
involved, and the dependency of the running time on~$k$ is hard to pin down exactly, but
at the very least multiply exponential. As for polynomial kernelization, both are
infeasible: \textsc{$k$-Path} and \textsc{$k$-Cycle} were among the first problems to
which the lower bounds framework was applied, and \textsc{Disjoint Paths} was addressed
in~\cite{BodlaenderTY11}. In both cases, the conclusion is that neither problem allows a
polynomial kernelization (or even a polynomial-time compression into size~$\p(k)$) unless
the polynomial hierarchy collapses. 

The \textsc{$K$-Cycle} problem, in turn, may intuitively seem to be closer in nature to
the latter problem than the former -- e.g., it is a terminal connectivity problem,
parameterized by the number of terminals, and there is no obvious relation between the
parameter and the size of the solution. Indeed, the problem can be solved via applications
of the \textsc{Disjoint Paths} algorithm, and Kawarabayashi  
solved the problem in time~$2^{2^{k^{10}}}\p(n)$ using graph minors-type graph structural
reasoning~\cite{Kawarabayashi08}. However, recently, Bj\"orklund et
al.~\cite{BjorklundHT12} solved~\textsc{$K$-Cycle} using 
an approach much closer to those of the cited~\textsc{$k$-Path} algorithms: they define a
large polynomial, which can be evaluated in~$2^k\cdot\p(n)$ time, and which, when evaluated
over a field of characteristic two, is non-zero if and only if the instance is positive.
The result then follows from an application of the Schwartz-Zippel lemma. (In fact, they
solved the more general variant of finding a \emph{shortest} $K$-cycle.)
For kernelization, the status of~\textsc{$K$-Cycle} is so far unknown, but there are
several factors -- the lack of a small witness, the status of the related problems given
above, the apparent difficulty of the problem -- which would suggest that the answer
should be negative (i.e., that~\textsc{$K$-Cycle} should have no polynomial kernel).
As the present paper shows, this conclusion may well be mistaken.

\subparagraph*{Our results.}
We give an alternative algebraic algorithm for the \textsc{$K$-cycle} problem,
also with a running time of~$\Oh^*(2^{O(k)})$, by encoding the problem into a variant of the
Tutte matrix. More concretely, given~$G$ and~$K$ we construct a matrix~$M_G$ over 
GF$(2^\ell)$, whose entries are polynomials, and show that~$G$ has a~$K$-cycle if
and only if the determinant polynomial of~$M_G$ contains a certain type of
term. Further minor modifications of the matrix yield an algorithm with running
time~$\Oh^*(2^k)$, and a matrix structure such that careful application of partial random
evaluation and Gaussian elimination can reduce~$M_G$ to a matrix~$A$ with total coding
length~$\Oh(k^3)$, such that it can be decided from the determinant polynomial of~$A$ 
whether~$G$ has a~$K$-cycle. All in all, this yields a randomized polynomial compression
of~\textsc{$K$-Cycle} into space~$\Oh(k^3)$. The construction, and all proofs, are simple,
and we need only basic arguments about determinants and cycle covers to complete them. 

We note that our approach so far fails to provide a polynomial kernel, in the strict
sense; the reason being that the output is an instance of a different problem 
(of deciding a particular property of~$\det A$) which is not known to be in \NP,
while a kernelization requires that the output is an instance of the same problem.
This is closely related to the issue of the \emph{witness size} required
for~\textsc{$K$-Cycle}; we are not aware of a witness for either \textsc{$K$-Cycle} or
for our artificial algebraic output problem, of size~$\p(k+\log n)$. 
We consider these results quite surprising. 

\subparagraph*{Related work.}
The Tutte matrix (see Section~\ref{sec:prel}) is a skew-symmetric matrix of
indeterminates, created from the adjacency matrix of a graph~$G$, which is non-singular if
and only if~$G$ has a perfect matching~\cite{Tutte47}. This can be used to determine the
size of a maximum matching in randomized time~$\Oh(n^\omega)$~\cite{Lovasz79,MuchaS04},
where~$\omega<2.3727$ is the matrix multiplication exponent~\cite{Vassilevska12,Stothers,CoppersmithW90}. 
Mucha and Sankowski~\cite{MuchaS04} showed how to find a maximum matching in the same time.
Geelen~\cite{Geelen00} gave a deterministic polynomial-time procedure which finds a
maximum rank evaluation of the Tutte matrix (which does not lead to a competitive
deterministic matching algorithm, but may be of interest for the general question of
removing randomness due to applications of Schwartz-Zippel). 
For non-algebraic algorithms for matching, Micali and Vazirani~\cite{MicaliV80} gave an
algorithm that finds a maximum matching in general graphs in time~$\Oh(m\sqrt{n})$. 

Algebraic FPT algorithms, beyond those cited for~\textsc{$k$-Path} above, have been used
by, e.g., Lokshtanov and Nederlof~\cite{LokshtanovN10} and Cygan et al.~\cite{CyganKPPW12}.
See also Nederlof's PhD thesis~\cite{Nederlof}. More specifically, algorithms based around
determinant computations have been used by Bj\"orklund~\cite{Bjorklund10a,Bjorklund10b}.
However, we argue that the approach of the present paper leads to significantly simpler
algorithms and correctness proofs than before. Algebraically based \emph{kernelizations},
in particular using tools of matroid theory, have been given
in~\cite{KratschW12a,KratschW12b}. Related to the present work, it is interesting to note
that the result of~\cite{KratschW12a} was a pure compression, albeit within \NP (the
problem \textsc{Odd Cycle Transversal} was encoded into matroid, represented by a matrix
of total coding length~$\p(k)$), 
while~\cite{KratschW12b} gave graph-based reduction rules, significantly broadening the
applicability of the tools. A similar improvement on the tools of the present work would be highly
interesting.

\subparagraph*{Organization.}
We review some basic definitions in the next section, then Section~\ref{sec:themeat}
gives, in turn, a very simple~$\Oh^*(4^k)$-time for \textsc{$K$-Cycle}; an improvement to
an~$\Oh^*(2^k)$-time algorithm; and the Gaussian polynomial compression. 

\section{Preliminaries} \label{sec:prel}

\subparagraph*{Parameterized complexity.}
A \emph{parameterized problem} is a language~$\Q\subseteq\Sigma^*\times\N$; the second component of instances~$(x,k)$ is called the parameter (cf.~\cite{DowneyF98}). A parameterized problem is \emph{fixed-parameter tractable} (FPT) if there is an algorithm~$A$ and a computable function~$f\colon\N\to\N$ such that~$A$ decides~$(x,k)\in\Q$ in time~$f(k)|x|^{\Oh(1)}$. A \emph{kernelization} \emph{of~$\Q$} is a polynomial-time computable mapping~$K\colon\Sigma^*\times\N\to\Sigma^*\times\N\colon(x,k)\mapsto(x',k')$ such that $(x,k)\in\Q$ if and only if~$(x',k')\in\Q$ and with~$|x'|,k'\leq h(k)$ where~$h$ is a computable function;~$h$ is called the \emph{size} of the kernel and~$K$ is a \emph{polynomial kernelization} if~$h(k)$ is polynomially bounded.
A \emph{polynomial compression} is a polynomial kernelization relaxed so that the output
may be an instance of a (fixed) different language than the input language. 
This has also been called \emph{bikernel}~\cite{AlonGKSY10} and \emph{generalized kernelization}~\cite{BodlaenderDFH09}).

\subparagraph*{The Tutte matrix.}
Let~$G=(V,E)$ be a simple undirected graph with~$V=\{v_1,\ldots,v_n\}$. The \emph{Tutte
  matrix}~$A_G$ is the~$n \times n$ matrix of indeterminates such that
\[
A_G(i,j) = 
\begin{cases} x_{ij} & \text{if $v_iv_j \in E$ and $i<j$,} \\ 
-x_{ji} & \text{if $v_iv_j \in E$ and $i>j$,} \\ 
0 & \text{otherwise,}
\end{cases}
\]
where~$x_{ij}$ are distinct commuting variables. Tutte~\cite{Tutte47} showed that~$\det
A_G\neq 0$ (viewed as a polynomial) if and only if~$G$ has a perfect matching. 
Lovasz~\cite{Lovasz79} showed the applications of this type of result to randomized
algorithms. 

\subparagraph*{Determinants and cycle covers.}
We recall a few basic facts.
Let~$D=(V,E)$ be a directed graph, which may contain loops. A \emph{cycle cover} of~$D$
is a set~$\C \subseteq E$ of arcs such that every vertex in~$D$ has in- and out-degree
exactly one in~$\C$. We allow loops to be present in the cycle cover. For an undirected
simple graph~$G$, which again may contain loops, an \emph{oriented cycle cover} of~$G$ is
a cycle cover of the bidirectional graph corresponding to~$G$. (Note that this implies
that loops and isolated edges are permitted in the cycle cover, corresponding to cycles of
length one respectively two.) 

Let~$A$ be an~$n \times n$ matrix over a field of characteristic two. Then the determinant
and the permanent of~$A$ coincide:
\begin{equation}
\label{eqn:det}
\det A = \perm A = \sum_{\sigma \in S_n} \prod_{i=1}^n A(i,\sigma(i)),
\end{equation}
where~$S_n$ is the set of all permutations of~$[n]$. Let~$D$ be a directed graph on vertex
set~$V=\{v_1,\ldots,v_n\}$ such that~$v_iv_j \in E(D)$ if and only if~$A(i,j)\neq 0$
(where~$v_iv_i$ denotes a loop on the vertex~$v_i$). There is a well-known bijection
between terms of the na\"ive summation~(\ref{eqn:det}) of the determinant and cycle covers
of~$D$, as follows.

\begin{proposition}
\label{prop:cyclecovers}
Let~$D=(V,E)$ and~$A$ be as above. For a permutation~$\sigma \in S_n$,
let~$\C_\sigma=\{v_iv_{\sigma(i)}: i \in [n]\}$. If~$\C_\sigma \subseteq E$, then~$\C_\sigma$
is a cycle cover of~$D$; furthermore, this describes a bijection between cycle covers
of~$D$ and non-zero terms in the summation~(\ref{eqn:det}). 
\end{proposition}

Let~$C$ be a simple cycle in a cycle cover~$\C_\sigma$. We call~$C$ \emph{reversible} if
the cycle has length at least three and for every edge~$v_iv_j \in C$, we
have~$A(i,j)=A(j,i)$; further, we call~$\C_\sigma$ reversible if it contains at least one
reversible cycle. A critical observation, both in previous and present work, is that
reversible cycle covers cancel in~(\ref{eqn:det}).

\begin{proposition}
\label{prop:reverse}
If~(\ref{eqn:det}) is computed over a field of characteristic two, then the terms
corresponding to reversible cycle covers cancel each other. 
\end{proposition}
\begin{proof}
Let~$\C_\sigma$ be a reversible cycle cover, and let~$C$ be the first reversible cycle
of~$\C_\sigma$, counted by vertex incidence (i.e., the cycles of~$\C_\sigma$ are sorted
according to the number of the earliest incident vertex). 
Let~$\C'=\C_{\sigma'}$ be the cycle cover resulting from reversing~$C$. Then this
operation creates a fix-point-free involution among the reversible cycle covers. Further,
as the terms of~(\ref{eqn:det}) corresponding to~$\sigma$ and~$\sigma'$ are identical by
definition, all terms of~(\ref{eqn:det}) corresponding to reversible cycle covers will
cancel each other out. 
\end{proof}

Thus, when reasoning about the surviving terms of~$\det A$, we only need to concern
ourselves with non-reversible cycle covers. (In particular, if~$A=A_G$ is the Tutte matrix
of a graph~$G$ over a field of characteristic two, then every cycle of length more than
two is reversible, and there are no cycles of length one; thus the non-reversible cycle
covers are exactly the perfect matchings of~$G$.) 

\subparagraph*{Schwartz-Zippel.} 
We also recall the Schwartz-Zippel lemma. 

\begin{lemma}[Schwartz-Zippel~\cite{Schwartz80,Zippel79}] Let~$P(x_1,\ldots,x_n)$ be a multivariate polynomial of
  total degree at most~$d$ over a field~$\F$, and assume that~$P$ is not identically zero. 
  Pick~$r_1, \ldots, r_n$ uniformly at random from~$\F$. Then~$\Pr (P(r_1,\ldots,r_n)=0)
  \leq d/|F|$. 
\end{lemma}

We will use this mostly, though not exclusively, for the case that~$P$ is the determinant
of a matrix over GF$(2^\ell)$. 

\subparagraph*{Detecting monomials in a polynomial.}
Our final generic ingredient is an application of inclusion-exclusion to finding certain
monomials in a polynomial over a field of characteristic two. For a polynomial~$P$ and a
monomial~$m$, we let~$P(m)$ denote the coefficient of~$m$ in~$P$. We need a way to extract
from~$P$ only those monomials divided by a certain term. 

\begin{lemma}
\label{lm:polypie}
Let~$P(x_1,\ldots,x_n)$ be a polynomial over a field of characteristic two, 
and~$T \subseteq [n]$ a set of target indices. For a set~$I \subseteq [n]$,
define~$P_{-I}(x_1,\ldots,x_n)=P(y_1,\ldots,y_n)$ where~$y_i=0$ for~$i \in I$ and~$y_i=x_i$
otherwise. Define 
\[
Q(x_1,\ldots,x_n) = \sum_{I \subseteq T} P_{-I}(x_1,\ldots,x_n).
\]
Then for any monomial~$m$ such that~$t:=\prod_{i \in T} x_i$ divides~$m$ we have~$Q(m)=P(m)$, 
and for every other monomial we have~$Q(m)=0$.
\end{lemma}
\begin{proof}
Consider a monomial~$m$ with non-zero coefficient in~$P$. Observe first that for every~$I
\subseteq [n]$, we have~$P_{-I}(m)=P(m)$ if no variable~$x_i$ with~$i \in I$ occurs in~$m$,
and~$P_{-I}(m)=0$ otherwise. Now, if~$t$ divides~$m$, then out of the~$2^{|T|}$ evaluations,
the monomial~$m$ occurs in exactly one (namely,~$I=\emptyset$). Thus,~$Q(m)=P(m)$. 
If~$t$ does not divide~$m$, let~$J=\{i \in I: x_i \text{ does not divide }m\}$,
and observe that~$P_{-I}(m)=P(m)$ for every~$I \subseteq J$. Since~$J \neq \emptyset$,
this is an even number of occurrences of the same monomial with the same coefficient,
which implies that they sum to zero. Applying this argument individually to every monomial
in~$P$ accounts for all occurrences of monomials in the sum defining~$Q$; the result
follows.
\end{proof}

We remark that we do not require~$P$ to be multilinear (although we do require~$T$ to be a
set rather than a multiset).

\section{An Algebraic FPT Algorithm}
\label{sec:themeat}

We now give our algorithm and compression for \textsc{$K$-Cycle}. Let us first fix a
definition. 

\begin{definition}
For a vertex $v \in V$, a $v$-cycle is a cycle that passes through $v$. For a set $T
\subseteq V$, a $T$-cycle is a cycle that passes through all vertices of $T$. In both
cases, the cycle may pass through further vertices, but this is not required.
\end{definition}

The problem is then formally defined as follows. 

\begin{quotation}
\textsc{$K$-Cycle}
\\ \noindent {\bf Input:} A graph $G=(V,E)$; a set $K \subseteq V$ of terminal vertices
\\ \noindent {\bf Parameter:} $k:=|K|$
\\ \noindent {\bf Question:} Is there a~$K$-cycle in $G$?
\end{quotation}

We will show two algebraic FPT algorithms for this problem, giving two ways of encoding it
into the determinant of a matrix. We then show how this implies a polynomial compression
via Gaussian elimination, into space $\Oh(k^3)$.

\subsection{Graph preprocessing}

We begin with a simple preprocessing of the graph (reducing the terminals to degree two).

\begin{lemma}
\label{lemma:degtwo}
Let $(G,K)$ be an instance of \textsc{$K$-Cycle} with $|K|>1$. We can reduce $(G,K)$ to an
equivalent instance $(G',K')$, $|K'|=|K|$, where $d(v)=2$ for every $v \in K'$, and where
$K'$ is an independent set with no common neighbours.
\end{lemma}
\begin{proof}
We assume that $K$ is an independent set in $G$ (by subdividing edges within $K$, 
if necessary). Construct $G'$ from $G$ by replacing every terminal $v \in K$ by two
non-adjacent copies $v', v''$ (with neighbourhoods identical to that of~$v$). Create
a new vertex $v$ with $N(v)=\{v',v''\}$. The new terminal set $K'$ consists of these new
vertices $v$. 

It is easy to show that this reduction maintains the solution status. On the one hand, for
any $K$-cycle in $G$, we may replace each portion $u-v-w$ of the cycle, with $v \in K$ and
hence $u, w \notin K$, by a path $u-v'-v-v''-w$, hitting the new terminal $v$. On the
other hand, any $K'$-cycle in $G'$ must pass through both neighbours $v', v''$ of each
terminal $v \in K'$, and these neighbours are distinct for all terminals. Thus if each
segment $v'-v-v''$ of the $K'$-cycle in $G'$ is contracted into $v$, we get a valid
$K$-cycle in $G$. 
\end{proof}

The requirement that $|K|>1$ comes from the consideration of whether a single edge $uv$
should be considered a $K$-cycle with $K=\{u\}$ (but the case~$|K|=1$ is in either case
easily solvable in polynomial time).

For the rest of the paper, for convenience, we will let $G=(V,E)$ be a graph, on vertex
set $V=\{v_1, \ldots, v_n\}$ and terminal set $K=\{v_1, \ldots, v_k\}$, to which the above
reduction has already been applied. We also assume $N(v_i)=\{v_{k+2i-1}, v_{k+2i}\}$ for
$i\in [k]$.

\subsection{Matrix construction}

We now show the matrix which will encode the existence of a~$K$-cycle. We begin with a
more intuitive construction, that implies a running time of~$\Oh^*(4^k)$, then modify it
to arrive at the~$\Oh^*(2^k)$-time algorithm and polynomial compression. 

Given a graph~$G$, reduced as per the previous subsection, we define the matrix~$A_G$ as
follows. We start from the Tutte matrix~$A_G$ of~$G$ (although, as the field is of
characteristic two, we will not observe the signs), and adjust so that~$A(i,i)=1$
for~$i>3k$ (effectively adding self-loops to all vertices except~$N[K]$). Finally, we
orient the edges incident to~$v_1$ to make~$v_1$-cycles non-reversible: let~$A(1,k+1)$
and~$A(k+2,1)$ be unmodified, but set~$A(1,k+2)=A(k+1,1)=0$. This can be done safely, as
any~$K$-cycle of~$G$ can be oriented in either direction. 

Let~$M_G$ denote the resulting matrix. We can detect a~$K$-cycle in~$G$ as follows.

\begin{theorem}
\label{th:fourtok}
Let~$T=\{x_{i,k+2i-1}, x_{i,k+2i}: i \in [k]\}$ and~$t=\prod_{x \in T} x$. Then~$G$ has
a~$K$-cycle if and only if~$\det M_G$, viewed as a polynomial, contains a monomial~$m$
with non-zero coefficient such that~$t$ divides~$m$. 
\end{theorem}
\begin{proof}
Recall the summation~(\ref{eqn:det}) and the notion of a reversible cycle from
Section~\ref{sec:prel}. We claim a one-to-one correspondence between non-zero monomials
of~$\det M_G$ and non-reversible cycle covers. 

This follows from basic observations, but we prove it for completeness. 
Since~(\ref{eqn:det}) is already in sum-product form, every non-zero monomial of~$\det
M_G$ corresponds to a non-empty set of summands from~(\ref{eqn:det}). By
Prop.~\ref{prop:reverse}, we get the same result if we restrict ourselves to those
summands corresponding to non-reversible cycle covers. We show that two summands,
corresponding to distinct non-reversible cycle covers~$\C, \C'$, always produce distinct
monomials: if~$\C$ and~$\C'$ use distinct sets of underlying undirected, non-loop edges,
then the claim is clear, and the set of loop edges of a cycle cover is a function of the
set of non-loop edges. In the remaining case,~$\C'$ must be attainable by a reorientation of~$\C$.
However, there are by construction only three types of non-reversible cycles: loops,
isolated edges, and~$v_1$-cycles, where a~$v_1$-cycle cannot be reversed, and loops and
isolated cycles are invariant under reversal. Thus~$\C$ and~$\C'$ must produce distinct
monomials. 

The result is now simple. First, if~$C$ is a~$K$-cycle in~$G$, then it contributes all
factors in~$t$, and by padding~$C$ using self-loops we produce a non-reversible cycle
cover, which produces a non-zero monomial of~$\det M_G$. On the other hand, if a non-zero
monomial in~$\det M_G$ contains the factor~$x_{i,k+2i-1}x_{i,k+2i}$, then in the
corresponding cycle cover, the~$v_1$-cycle most also pass through~$v_i$, as such a
factor cannot be contributed by loops and isolated edges.  By induction, if a non-zero
monomial in~$\det M_G$ is divided by~$t$, then the corresponding cycle cover contains
a~$v_1$-cycle which passes through every vertex of~$K$, i.e., a~$K$-cycle. 
\end{proof}

As~$|T|=2k$, this implies an~$\Oh(2^{2k})$-time randomized algorithm for the problem,
via~Lemma~\ref{lm:polypie} and by evaluating the resulting polynomial~$Q$ randomly over
GF$(2^\ell)$ for~$\ell=\Omega(\log n)$. We will improve this in two ways: by
introducing a modification which will let us match the~$\Oh^*(2^k)$ running time
of~\cite{BjorklundHT12}, and by showing how to use Gaussian elimination and partial random
evaluation to produce a polynomial compression. 

\subsection{A $\mathbf{2^k}$ Algorithm}

We now show a different way to determine the existence of a~$K$-cycle from~$M_G$. 
Let an \emph{orientation of~$M_G$} be the result of, for every~$v_i \in K$, $i>1$, either
setting~$A(k+2i-1,i)=A(i,k+2i)=0$ or~$A(k+2i,i)=A(i,k+2i-1)=0$, i.e., orienting the edges
incident to~$v_i$ either as~$v_i' \rightarrow v_i \rightarrow v_i''$ or as~$v_i'
\leftarrow v_i \leftarrow v_i''$. We claim the following.

\begin{theorem}
\label{th:reorient}
Let~$Q'$ be the sum of~$\det M_G'$ over all~$2^{k-1}$ orientations~$M_G'$ of~$M_G$. 
Then~$G$ has a~$K$-cycle if and only if~$Q'$ is not identically zero.
\end{theorem}
\begin{proof}
Let~$M_G'$ be an arbitrary orientation of~$M_G$. As in Theorem~\ref{th:fourtok}, monomials
of~$\det M_G'$ correspond to non-reversible cycle covers, but now, every cycle incident on
some~$v_i \in K$ counts as non-reversible (and again, attempting to reverse such a cycle
produces a zero-term). On the other hand, if two orientations~$M_G'$ and~$M_G''$ contain
non-reversible cycle covers~$\C'$ and~$\C''$ such that~$\C''$ can be obtained by
reorienting~$\C'$, then~$\C'$ and~$\C''$ contribute identical monomials to the sum, and
their contribution may cancel. Thus, let~$\C^*$ be an \emph{unoriented} cycle cover, such
that every cycle in~$\C^*$ is either a loop, an isolated edge, or a cycle incident on~$K$,
and such that~$K$ is covered entirely by the latter type of cycles. We will count the
number of contributions of orientations of~$\C^*$ to the sum. 

For this, simply observe that in a single cycle~$C$ of~$\C^*$, as soon as the orientation
of at least one vertex of~$C$ has been determined, the direction taken through every other
vertex of~$C$ is fixed as a consequence. Thus, if~$\C^*$ contains a~$K$-cycle~$C$, then
only one orientation~$M_G'$ is possible, as the vertex~$v_1$ enforces a direction already
in~$M_G$. On the other hand, if~$\C^*$ contains at least two cycles incident on~$K$, then
all cycles \emph{not} incident on~$v_1$ may be oriented arbitrarily, making for an even
number of orientations, each one of which contributes the same monomial to the sum. 

Thus non-zero monomials of~$Q'$ correspond to~$K$-cycles in~$G$, as promised. 
\end{proof}

This construction brings our algorithm closer in spirit to the determinant sums of
Bj\"orklund~\cite{Bjorklund10a,Bjorklund10b}, or the algebraic FPT algorithms of Cygan et
al.~\cite{CyganNPPvRW11}. However, as the next subsection shows, by bringing the algorithm
back into the structure of deciding properties of the determinant polynomial of a single
matrix, we get a randomized polynomial compression for~\textsc{$K$-Cycle} via Gaussian
elimination. 

\subsection{Polynomial Compression}

Now, we finally show how to use the above for a polynomial compression of
the~\textsc{$K$-Cycle} problem. 

We describe one final modification of the matrix~$M_G$. For every~$v_i \in K$,~$i>1$, we
introduce a new variable~$a_i$, and multiply~$A(k+2i-1,i)$ and~$A(i,k+2i)$
by~$a_i$, and~$A(k+2i,i)$ and~$A(i,k+2i-1)$ by~$1-a_i$. Observe that this implies that the
algorithm of Theorem~\ref{th:reorient} can be executed by iteratively setting each~$a_i$
to either~$1$ or~$0$, and computing the determinant each time. Strictly speaking,
each individual determinant computation would then seem to require a fresh dose of
randomness, via the Schwartz-Zippel evaluation step, making the approach inappropriate for
kernelization. We show that it is possible to perform this in the alternate direction,
first randomly evaluating every variable~$x_e$ for~$e \in E(G)$, then performing Gaussian
elimination into a compressed output, and finally (at some future time) performing
the~$2^{k-1}$ assignments to the variables~$a_i$ and computing the resulting
determinants. 

\begin{theorem}
The \textsc{$K$-Cycle} problem has a randomized polynomial compression of size~$\Oh(k^3)$.
\end{theorem}
\begin{proof}
We get the result in two steps, first showing that we can randomly evaluate the
variables~$\mathbf{x}$ while leaving~$\mathbf{a}$ as indeterminates,
then applying Gaussian elimination to produce a smaller matrix with the same determinant
(viewed as a polynomial in~$\mathbf{a}$). 

Let~$P(\mathbf{x}, \mathbf{a})$ be the determinant polynomial of~$M_G$.
Define~$Q(\mathbf{x})$ to be the sum over the~$2^{k-1}$ instantiations of~$\mathbf{a}$
necessary to emulate the algorithm of Theorem~\ref{th:reorient};
observe that~$Q(\mathbf{x})$ is a polynomial of degree~$n$, and that~$Q(\mathbf{x})$ is
identically zero if and only if~$G$ has no~$K$-cycle. Thus, again by Schwartz-Zippel,
we may instantiate~$\mathbf{x}$ randomly from GF$(2^\ell)$, and with probability at
least~$n/2^\ell$ the resulting values are such that the~$2^{k-1}$-sized evaluation
of~$Q(\mathbf{x})$ would return non-zero. Picking~$\ell=\Theta(\log n)$ is sufficient for
this step to succeed with polynomial probability in~$n$ (and with~$\ell=\Theta(\log n+k)$, 
we get a failure rate still polynomial in~$n$, but exponentially small in~$k$). 
Note that by standard observations we may assume~$k \geq \log n$, as otherwise the
$\Oh^*(2^k)$-algorithm runs in polynomial time. 

Now, observe that we only need to know the existence of the polynomial~$Q$ for the above
correctness argument. Thus, by replacing~$\mathbf{x}$ randomly by values from
GF$(2^{\ell})$, we get a matrix~$M'$ with mostly concrete values, and indeterminates in
the top-left~$3k \times 3k$ corner, such that preserving~$\det M'$ is sufficient (up to
the failure probability in the previous step) for preserving the information of
whether~$G$ has a~$K$-cycle.

Next, recall that row and column operations preserve the determinant of a matrix exactly. 
We show that we can reduce~$M'$ to a blocks form
\[
M' = \left(
\begin{array}{cc}
A & 0 \\
0 & C \\
\end{array}
\right),
\]
where~$C$ is a matrix without indeterminates. Thus we will have~$\det M'=(\det A)(\det C)$ 
where~$\det C$ is a constant. 

This is easy. For sets~$R, C \subseteq [n]$, let~$M[R,C]$ denote the induced submatrix
of~$M$ with rows~$R$ and columns~$C$. 
Observe that the submatrix~$M_G[[3k+1,n],[3k+1,n]]$ is non-singular, as the diagonal
contributes the term~$1$ to the determinant and every other term will contain at least one
indeterminate. Thus (up to the failure probability),~$M'[[3k+1,n],[3k+1,n]]$ is
non-singular, and can be reduced to diagonal form with a non-zero diagonal, without
introducing any new indeterminate entries in~$M'$. Now we can use further row and column
operations to reduce~$M'[[1,3k],[3k+1,n]]$ and~$M'[[3k+1,n],[1,3k]]$ to all-zero matrices
(thereby modifying the contents of~$M'[[1,3k],[1,3k]]$, but not~$M'[[3k+1,n],[3k+1,n]]$).
This creates the desired blocks form, and every step preserves the determinant precisely
and is performed without further failure probability or growth of the individual entries
(since we are working over a finite field). 

Finally, we consider the resulting contents of the matrix~$A$. Initially, the entries
of~$M'(i,j)$ for~$i,j \leq 3k$ are either constants, or expressions~$a_i \cdot c + c'$ for
some constants~$c, c'$.  Every further row or column operation that modifies 
these entries adds some concrete value~$c'$ to the entry, meaning that we can maintain
these entries in the form~$a_i\cdot c+c'$ where~$c, c'$ are concrete values from
GF$(2^\ell)$; thus the coding length remains~$\Oh(\ell)$ bits per entry.
We then multiply one arbitrary row of~$A$ by~$\det C$, which again only has the effect 
of modifying the values~$c, c'$. This gives us a~$3k \times 3k$ matrix~$A'$, with entries
encoded into~$\Oh(\ell)=\Oh(k)$ bits, such that~$\det A'=\det M'$, where~$M'$ is the
matrix produced by randomly instantiating~$\mathbf{x}$ in~$M_G$. 
\end{proof}

Finally, we remark that, unusually, the output problem is not trivially in \NP (as it is a
question about the outcome of an exponentially large computation). Thus in terms of
parameterized complexity, we do not strictly speaking get a polynomial kernel, as we know
of no way of getting back from the matrix~$A'$ above to an instance of~\textsc{$K$-Cycle}.

\section{Conclusions}

We have shown an alternate algebraic algorithm for the~\textsc{$K$-Cycle} problem,
recasting the original problem into a question about the existence of certain terms in the
determinant polynomial of a matrix with indeterminate entries. By careful application of
partial evaluation and Gaussian elimination, we have shown that this leads to a
polynomial compression of a~\textsc{$K$-Cycle} instance into space~$\Oh(|K|^3)$.
This partially answers the question of the kernelizability of~\textsc{$K$-Cycle}, in a
perhaps surprising direction. 

Although we are not able to produce a proper kernel, since we are not able to get back to
an instance of the~\textsc{$K$-Cycle} problem, such kernel-like polynomial compressions
have been previously considered in parameterized complexity~\cite{AlonGKSY10}, and in fact
all existing frameworks for excluding polynomial kernelization
(e.g.,~\cite{FortnowS11,Drucker12}) also exclude polynomial compressions. 
Thus, for the sake of a smooth theory, we hope that the~\textsc{$K$-Cycle} problem can
also be shown to have a polynomial kernel (e.g., a compression within \NP). 

Another interesting improvement would be a more direct kernel, e.g., based on reduction
rules which make direct modifications to~$G$ and~$K$. The tools required for finding such
rules may well have further applications (perhaps analogously to the two previous
works~\cite{KratschW12a,KratschW12b}). 

It would also be interesting to consider further related problems, perhaps starting with
the problems of finding a \emph{shortest}~$K$-cycle, and a~$K$-cycle with a prescribed
parity, as these problems can also be solved by the approach in~\cite{BjorklundHT12}.
While it seems that our algorithm can be adapted for this setting, it is not clear to us
at the moment whether this can be done in a way that allows for a polynomial compression.

\subparagraph*{Acknowledgements}
The author is grateful to Thore Husfeldt and Stefan Kratsch for rewarding discussions,
and to an anonymous reviewer for suggesting improvements to the paper.

\bibliographystyle{abbrv}
\bibliography{ksetcycle}

\begin{thebibliography}{10}

\bibitem{AlonGKSY10}
N.~Alon, G.~Gutin, E.~Kim, S.~Szeider, and A.~Yeo.
\newblock Solving {MAX-$r$-SAT} above a tight lower bound.
\newblock {\em Algorithmica}, pages 1--18, 2010.

\bibitem{AlonYZ95}
N.~Alon, R.~Yuster, and U.~Zwick.
\newblock Color-coding.
\newblock {\em J. ACM}, 42(4):844--856, 1995.

\bibitem{Bjorklund10b}
A.~Bj{\"o}rklund.
\newblock Determinant sums for undirected hamiltonicity.
\newblock In {\em FOCS}, pages 173--182, 2010.

\bibitem{Bjorklund10a}
A.~Bj{\"o}rklund.
\newblock Exact covers via determinants.
\newblock In {\em STACS}, pages 95--106, 2010.

\bibitem{BjorklundHKK10}
A.~Bj{\"o}rklund, T.~Husfeldt, P.~Kaski, and M.~Koivisto.
\newblock Narrow sieves for parameterized paths and packings.
\newblock {\em CoRR}, arXiv:1007.1161, 2010.

\bibitem{BjorklundHT12}
A.~Bj{\"o}rklund, T.~Husfeldt, and N.~Taslaman.
\newblock Shortest cycle through specified elements.
\newblock In {\em SODA}, pages 1747--1753, 2012.

\bibitem{FellowsFestschrift}
H.~L. Bodlaender, R.~Downey, F.~V. Fomin, and D.~Marx, editors.
\newblock {\em The Multivariate Algorithmic Revolution and Beyond - Essays
  Dedicated to Michael R. Fellows on the Occasion of His 60th Birthday}, volume
  7370 of {\em Lecture Notes in Computer Science}. Springer, 2012.

\bibitem{BodlaenderDFH09}
H.~L. Bodlaender, R.~G. Downey, M.~R. Fellows, and D.~Hermelin.
\newblock On problems without polynomial kernels.
\newblock {\em J. Comput. Syst. Sci.}, 75(8):423--434, 2009.

\bibitem{BodlaenderFLPST09}
H.~L. Bodlaender, F.~V. Fomin, D.~Lokshtanov, E.~Penninkx, S.~Saurabh, and
  D.~M. Thilikos.
\newblock ({M}eta) kernelization.
\newblock In {\em FOCS}, pages 629--638, 2009.

\bibitem{BodlaenderJK11}
H.~L. Bodlaender, B.~M.~P. Jansen, and S.~Kratsch.
\newblock Cross-composition: A new technique for kernelization lower bounds.
\newblock In {\em STACS}, pages 165--176, 2011.

\bibitem{BodlaenderTY11}
H.~L. Bodlaender, S.~Thomass{\'e}, and A.~Yeo.
\newblock Kernel bounds for disjoint cycles and disjoint paths.
\newblock {\em Theor. Comput. Sci.}, 412(35):4570--4578, 2011.

\bibitem{ChenKX10}
J.~Chen, I.~A. Kanj, and G.~Xia.
\newblock Improved upper bounds for vertex cover.
\newblock {\em Theor. Comput. Sci.}, 411(40-42):3736--3756, 2010.

\bibitem{CoppersmithW90}
D.~Coppersmith and S.~Winograd.
\newblock Matrix multiplication via arithmetic progressions.
\newblock {\em J. Symb. Comput.}, 9(3):251--280, 1990.

\bibitem{CyganKPPW12}
M.~Cygan, S.~Kratsch, M.~Pilipczuk, M.~Pilipczuk, and M.~Wahlstr{\"o}m.
\newblock Clique cover and graph separation: New incompressibility results.
\newblock In {\em ICALP (1)}, pages 254--265, 2012.

\bibitem{CyganNPPvRW11}
M.~Cygan, J.~Nederlof, M.~Pilipczuk, M.~Pilipczuk, J.~M.~M. van Rooij, and
  J.~O. Wojtaszczyk.
\newblock Solving connectivity problems parameterized by treewidth in single
  exponential time.
\newblock In {\em FOCS}, pages 150--159, 2011.

\bibitem{DellM12}
H.~Dell and D.~Marx.
\newblock Kernelization of packing problems.
\newblock In {\em SODA}, pages 68--81, 2012.

\bibitem{DellvM10}
H.~Dell and D.~van Melkebeek.
\newblock Satisfiability allows no nontrivial sparsification unless the
  polynomial-time hierarchy collapses.
\newblock In {\em STOC}, pages 251--260, 2010.

\bibitem{DomLS09}
M.~Dom, D.~Lokshtanov, and S.~Saurabh.
\newblock Incompressibility through colors and {ID}s.
\newblock In {\em ICALP (1)}, pages 378--389, 2009.

\bibitem{DowneyF98}
R.~G. Downey and M.~R. Fellows.
\newblock {\em Parameterized Complexity}.
\newblock Springer, November 1998.

\bibitem{Drucker12}
A.~Drucker.
\newblock New limits to classical and quantum instance compression.
\newblock In {\em FOCS}, pages 609--618, 2012.

\bibitem{FlumG06}
J.~Flum and M.~Grohe.
\newblock {\em Parameterized Complexity Theory}.
\newblock Springer, March 2006.

\bibitem{FominLST10}
F.~V. Fomin, D.~Lokshtanov, S.~Saurabh, and D.~M. Thilikos.
\newblock Bidimensionality and kernels.
\newblock In M.~Charikar, editor, {\em SODA}, pages 503--510. SIAM, 2010.

\bibitem{FortnowS11}
L.~Fortnow and R.~Santhanam.
\newblock Infeasibility of instance compression and succinct {PCP}s for {NP}.
\newblock {\em J. Comput. Syst. Sci.}, 77(1):91--106, 2011.

\bibitem{Geelen00}
J.~F. Geelen.
\newblock An algebraic matching algorithm.
\newblock {\em Combinatorica}, 20(1):61--70, 2000.

\bibitem{HermelinW12}
D.~Hermelin and X.~Wu.
\newblock Weak compositions and their applications to polynomial lower bounds
  for kernelization.
\newblock In {\em SODA}, pages 104--113, 2012.

\bibitem{Kawarabayashi08}
K.~Kawarabayashi.
\newblock An improved algorithm for finding cycles through elements.
\newblock In {\em IPCO}, pages 374--384, 2008.

\bibitem{KawarabayashiKR12}
K.~Kawarabayashi, Y.~Kobayashi, and B.~A. Reed.
\newblock The disjoint paths problem in quadratic time.
\newblock {\em J. Comb. Theory, Ser. B}, 102(2):424--435, 2012.

\bibitem{Koutis08}
I.~Koutis.
\newblock Faster algebraic algorithms for path and packing problems.
\newblock In {\em ICALP (1)}, pages 575--586, 2008.

\bibitem{KratschW12a}
S.~Kratsch and M.~Wahlstr{\"o}m.
\newblock Compression via matroids: a randomized polynomial kernel for odd
  cycle transversal.
\newblock In {\em SODA}, pages 94--103, 2012.

\bibitem{KratschW12b}
S.~Kratsch and M.~Wahlstr{\"o}m.
\newblock Representative sets and irrelevant vertices: new tools for
  kernelization.
\newblock In {\em FOCS}, pages 450--459, 2012.

\bibitem{LokshtanovN10}
D.~Lokshtanov and J.~Nederlof.
\newblock Saving space by algebraization.
\newblock In {\em STOC}, pages 321--330, 2010.

\bibitem{Lovasz79}
L.~Lov{\'a}sz.
\newblock On determinants, matchings, and random algorithms.
\newblock In {\em FCT}, pages 565--574, 1979.

\bibitem{MicaliV80}
S.~Micali and V.~V. Vazirani.
\newblock An {$O(\sqrt{|V|} |E|)$} algorithm for finding maximum matching in
  general graphs.
\newblock In {\em FOCS}, pages 17--27, 1980.

\bibitem{MuchaS04}
M.~Mucha and P.~Sankowski.
\newblock Maximum matchings via {G}aussian elimination.
\newblock In {\em FOCS}, pages 248--255, 2004.

\bibitem{Nederlof}
J.~Nederlof.
\newblock {\em Space and Time Efficient Structural Improvements of Dynamic
  Programming Algorithms}.
\newblock PhD thesis, University of Bergen, Norway, 2011.
\newblock Available at
  {\texttt{http://folk.uib.no/jne061/PhDthesisJesper.pdf}}.

\bibitem{NemhauserT75}
G.~Nemhauser and L.~Trotter.
\newblock Vertex packing: structural properties and algorithms.
\newblock {\em Mathematical Programming}, 8:232--248, 1975.

\bibitem{RobertsonS95}
N.~Robertson and P.~D. Seymour.
\newblock Graph minors. {XIII}. {T}he disjoint paths problem.
\newblock {\em J. Comb. Theory, Ser. B}, 63(1):65--110, 1995.

\bibitem{Schwartz80}
J.~T. Schwartz.
\newblock Fast probabilistic algorithms for verification of polynomial
  identities.
\newblock {\em J. ACM}, 27(4):701--717, 1980.

\bibitem{Stothers}
A.~Stothers.
\newblock {\em On the Complexity of Matrix Multiplication}.
\newblock PhD thesis, University of Edinburgh, 2010.

\bibitem{Thomasse10}
S.~Thomass{\'e}.
\newblock A 4{\it k}$^{\mbox{2}}$ kernel for feedback vertex set.
\newblock {\em ACM Transactions on Algorithms}, 6(2), 2010.

\bibitem{Tutte47}
W.~T. Tutte.
\newblock The factorization of linear graphs.
\newblock {\em J. London Math. Soc.}, s1-22(2):107--111, 1947.

\bibitem{Vassilevska12}
V.~Vassilevska~Williams.
\newblock Multiplying matrices faster than {C}oppersmith-{W}inograd.
\newblock In {\em STOC}, pages 887--898, 2012.

\bibitem{Williams09}
R.~Williams.
\newblock Finding paths of length {$k$} in {$O^*(2^k)$} time.
\newblock {\em Inf. Process. Lett.}, 109(6):315--318, 2009.

\bibitem{Zippel79}
R.~E. Zippel.
\newblock Probabilistic algorithms for sparse polynomials.
\newblock In {\em Symbolic and Algebraic Computation (EUROSAM)}, pages
  216--226, 1979.

\end{thebibliography}

\end{document}